\documentclass[11pt, oneside]{article}   	
\usepackage[]{geometry}

\usepackage{amsmath} 
\usepackage{tcolorbox}
\usepackage{color}
\usepackage[]{graphicx}
\usepackage{multirow}
\usepackage{graphicx}
\usepackage{array,tabularx}
\usepackage{booktabs}
\usepackage{amsmath}
\usepackage{hyperref}
\usepackage{amssymb}
\usepackage{amsmath}
\usepackage{amsthm}
\usepackage{color}
\usepackage{algorithm}
\usepackage{algpseudocode}
\usepackage{caption}
\usepackage{subcaption}

\newcommand\ddfrac[2]{\frac{\displaystyle #1}{\displaystyle #2}}

\newtheorem{lemma}{ Lemma}[section]
\newtheorem{proposition}{ Proposition}[section]

\newtheorem{remark}{ Remark}

\newtheorem{theorem}{ Theorem}[section]

\newtheorem{definition}{Definition}[section]
\newtheorem{notation}{Notation}[section]
\newtheorem{example}{Example}[section]

\DeclareMathOperator*{\argmax}{arg\,max}

\title{Partial-Information Q-Learning for General Two-Player Stochastic Games}
\author{\thanks{This work was partially supported by NSF grant DMS-1907518.}N. Medhin\thanks{Department of Mathematics, College of Sciences, North Carolina State University, Raleigh, NC, 27695
{\tt\small ngmedhin@ncsu.edu}}, A. Papanicolaou\thanks{Department of Mathematics, College of Sciences, North Carolina State University, Raleigh, NC, 27695
{\tt\small apapani@ncsu.edu}}, M. Zrida\thanks{Department of Operations Research, College of Sciences, North Carolina State University, Raleigh, NC, 27695
{\tt\small mzrida@ncsu.edu}}, 
}

\begin{document}

\maketitle 
\begin{abstract}
    In this article we analyze a partial-information Nash Q-learning algorithm for a general 2-player stochastic game. Partial information refers to the setting where a player does not know the strategy or the actions taken by the opposing player. We prove convergence of this partially informed algorithm for general 2-player games with finitely many states and actions, and we confirm that the limiting strategy is in fact a full-information Nash equilibrium. In implementation, partial information offers simplicity because it avoids computation of Nash equilibria at every time step. In contrast, full-information Q-learning uses the Lemke-Howson algorithm to compute Nash equilibria at every time step, which can be an effective approach but requires several assumptions to prove convergence and may have runtime error if Lemke-Howson encounters degeneracy. In simulations, the partial information results we obtain are comparable to those for full-information Q-learning and fictitious play. 
\end{abstract}

\section{Introduction}

We consider a general 2-player stochastic game where players learn their respective strategies over repeated rounds. Learning is conducted under partial information, which means a player does not observe the other player's actions or strategy. The only information available for learning is a commonly observed state variable, and of course a player will observe her own actions. Learning in the imperfect-information setting of \cite{heinrich2016deep} was shown to allow for improved outcomes in LeDuc and Limit Texas Hold'em poker. Compared to full information, partially informed learning is easier to implement because players do not need to find best responses to the strategies of other players. That is, a player computes a strategy that is an optimal response to their respective set of observations. Additionally, these limiting partial-information strategies are in fact a full-information Nash equilibrium, which we will prove in this paper.

The game that we consider has a finite space for the state variable, and finitely many possible actions that can be taken by the players. Evaluation of strategies can be done with a straight-forward tableau in a manner similar to Q-learning. For standard Markovian optimization problems there are Bellman equations, for which Q-learning is an effective estimation because there is usually a contraction mapping \cite{rust2008dynamic,sutton2018reinforcement}. The Q-functions for stochastic games have recursive equations similar to Bellman equations, but proofs for convergence and uniqueness are considerably more involved because there is not an easily identifiable contraction operator, except for special cases like zero-sum, symmetric or cooperative games \cite{hu1998multiagent,hu2003nash,littman2001friend,littman2001value,szepesvari1999unified,wang2002reinforcement}. For this reason, it is challenging to design a reinforcement learning algorithm that converges to a general Nash equilibrium. 

The full-information Nash Q-learning algorithm in \cite{hu1998multiagent} is proven to converge provided that intermittent Q-functions have either a global Nash equilibrium or a Nash saddle point, but such conditions are difficult to ensure. Additionally, full-information Q-learning requires computation of a joint Nash strategy at each iteration, either using the Lemke-Howson algorithm \cite{lemke1964equilibrium} or a fictitious-play approach \cite{brown1951iterative}. Repeatedly calling the Lemke-Howson algorithm can be slow; fictitious play can be a viable alternative but often has difficulty converging \cite{fudenberg1993learning,shamma2004unified}. 

The main result of this paper is a theorem proving convergence of partial-information Q-learning for general 2-player games having finite-action spaces, finite state space, and bounded reward functions. We are also able to confirm that these limiting strategies form a full-information Nash equilibrium. A similar partially informed learning for a min-max game was considered in \cite{kozuno2021model}, for which they showed regret bounds. For demonstration, we implement a deep neural network adaption of our partial-information algorithm to the Gridworld game (see \cite{hu1998multiagent}) and to LeDuc poker (see \cite{heinrich2016deep}), and we find outcomes that are comparable to those in the literature. 

The rest of the paper is organized as follows: Section \ref{sec:main} introduces the notation and definitions for our stochastic game; Section \ref{sec:partialInfo} introduces partial information, gives some conceptual examples, and proves the convergence of the partial-information algorithm; Section \ref{sec:simulations} presents simulation examples of (a) a randomly generated 2-player game (b) the Gridworld game learned with a Deep Q-Network (DQN) adaption of our algorithm, and (c) Led'uc Holdem learned also with the DQN adaption.

\section{Formulation of Stochastic Game}
\label{sec:main}
The actions available to Player 1 are a finite set denoted $A^1$, and the actions available to Player 2 are a finite set denoted $A^2$. We denote time with $t$, and we denote Players 1 and 2s' time-$t$ actions as $a_t^1\in A^1$ and $a_t^2\in A^2$, respectively. At a given time $t$ there is a state vector $s_t$ that takes value in a finite space $S$. This state vector make transitions in $S$ that are affected by the actions of both players. This transition distribution is expressed as
\[p(s'|s,a,a') :=\mathbb P(s_{t+1} = s'|s_t = s,a_t^1=a,a_t^2=a')\ ,\]
for any $a\in A^1$ and and $a'\in A^2$. Given state vector $s\in S$, the reward to Player 1 for taking action $a^1\in A^1$ when player 2 takes action $a^2\in A^2$, is $r^1(s,a^1,a^2)$. The reward for Player 2 in the same situation is $r^2(s,a^1,a^2)$. 

\begin{definition}
    \label{def:twoPlayerGame}
    We define our game by a tuple $\langle S, A^1, A^2, r^1, r^2, p \rangle$ where $S$ is respectively the game's state space, $A^1$ and $A^2$ are respectively player 1 and player 2 action spaces, $r^1, r^2: S \times A^1 \times A^2 \rightarrow \mathbb R^1$ are the payoff functions for our two players and $p: S \times A^1 \times A^2 \rightarrow \mathcal P(S)$ is the transition probability map with $\mathcal P(S)$ the set of probability distributions over $S$.
\end{definition}

\begin{definition}[Set of Distribution Vectors on a Set]
    \label{def:dist_vectors}
    For a finite set $S$ we let $\mathcal P(S)$ denote the space of probability vectors on $S$. For example, if $p\in\mathcal P(S)$, then $p(s)\in[0,1]$ and $\sum_{s\in S}p(s) = 1$.
\end{definition}

Let $\pi^i$ denote the strategy of Player $i$ for $i=1,2$. A strategy is a mapping from the state to a distribution on $i$’s actions,
\[s \mapsto \pi^i(s)\in  \mathcal P(A^i)\ ,\]
for any $s\in S$ with $\mathcal P(A^i)$ being the set of probability distribution vectors on $A^i$ as described in Definition \ref{def:dist_vectors}. For a given state $s\in S$, $\pi^i(s)$ is a vector indexed by the elements of $A^i$, which we denote as $\pi^i(s)=\left(\pi^i(s,a)\right)_{a\in A^i}$, and is a probability distribution such that 
\[\pi^i(s,a) =\mathbb P(a_t^i=a|s_t=s)\ ,\]
for all $a\in A^i$. That is, Player $i$'s time-$t$ action $a_t^i$, conditional on $s_t$, will be a random draw from $A^i$ according to the probability distribution $\pi^i(s_t)$, which we denote as
\[a_t^i\sim \pi^i(s_t,\cdot)\ .\]
If $\pi^i(s_t)$ is a point mass at some $a\in A^i$ then the strategy is a deterministic function of the state. Given strategies $\pi^1$ and $\pi^2$, the value functions of Players 1 and 2 are, 
\begin{align}
    \label{eq:vi}  v^i(s,\pi^1,\pi^2)&=\sum_{t=0}^\infty\gamma_i^t\mathbb E[r^i(s_t,a_t^1,a_t^2)|\pi^1,\pi^2,s_0=s]\quad\hbox{for }i=1,2\ ,
\end{align}
where $\gamma_i\in(0,1)$ are the players' discount factors. The probabilistic flow that produces the actions and state process nn \eqref{eq:vi} can summarized as follows,

\[\boxed{\begin{array}{c}\\s_{t}\sim p(\cdot|s_{t-1},a_{t-1}^1,a_{t-1}^2)\\
\\\end{array}}~~~\Rightarrow~~~ \boxed{\begin{array}{c}
a_t^1\sim\pi^1(s_t,\cdot)\\
a_t^2\sim\pi^2(s_t,\cdot)
\end{array}}
~~~~\Rightarrow~~~ \boxed{\begin{array}{c}\\s_{t+1}\sim p(\cdot|s_t,a_t^1,a_t^2)\\
\\\end{array}}\ , \]
where ``$\Rightarrow $" here means that the values on the left feed into the conditioning for random draws to the right. Our assumption throughout will be that, given $s_t=s$, each $a_t^i$ is an independent draw from $\pi^i(s,\cdot)$ for $i=1,2$.
\subsection{Nash Equilibria}
\label{sec:nashEquibriaDefs}
Let $\Pi^i$ for $i=1,2$ denote the space of possible strategies for the players. 

\begin{definition}[Nash Equilibrium]
    \label{def:Nash}
    A Nash equilibrium is a pair $(\pi_*^1,\pi_*^2)\in\Pi^1\times\Pi^2$ such that for all $s\in S$, 
\begin{align}
    \label{eq:NashEquilibrium}
    v^1(s,\pi_*^1,\pi_*^2)&\geq v^1(s,\pi^1,\pi_*^2)\quad\forall \pi^1\in\Pi^1\\
    \nonumber
    v^2(s,\pi_*^1,\pi_*^2)&\geq v^2(s,\pi_*^1,\pi^2)\quad\forall \pi^2\in\Pi^2\ ,
\end{align}
where $v^1$ and $v^2$ are the values functions given in \eqref{eq:vi}.
\end{definition}

In some sense, each player seeks a strategy to maximize the expectation of their respective reward. However, such an optimization is complicated because the opposing player's optimization will respond. If players' actions are rational and fully informed (i.e., both players know the other player's' strategy) then Player 1's strategy will be a best response to the strategy of Player 2, and similarly Player 2's strategy will be a best response to Player 1's strategy. 
\begin{definition}[Nash Q-Functions]
    \label{def:Qfunc}
    For any Nash equilibrium $(\pi_*^1,\pi_*^2)\in \Pi^1\times\Pi^2$, Player $i$'s Q-function is the reward received for a pair of actions played plus the future rewards when both players follow $(\pi_*^1,\pi_*^2)$, 
\begin{align}
    \label{eq:Qi}    Q_*^i(s,a^1,a^2)&=r^i(s,a^1,a^2)+\gamma_i\sum_{s'\in S}v^i(s',\pi_*^1,\pi_*^2)p(s'|s,a^1,a^2)\ ,
\end{align}
for all pairs $(a^1,a^2)\in A^1\times A^2$. 
\end{definition}

From equation \eqref{eq:Qi} we want to obtain a a tuple $\left<Q_*^1,Q_*^2,\pi_*^1,\pi_*^2\right>$ However, there is non-uniqueness because it is quite likely that there are multiple Nash equilibria, which means that a $Q_*^i$ for equation \eqref{eq:Qi} is non-unique, i.e., if there is a change in criterion for selection of Nash points then the ensuing $Q_*^i$ will change. Moreover, there may not be any favorable criteria for selection of Nash equilibria, which leads to some uncertainty in how to compute $Q_*^i$ from equation \eqref{eq:Qi} (this uncertainty will be discussed Section \ref{sec:uncertainty}). 

Nonetheless, equation \eqref{eq:Qi} remains an important tool in computation, and is useful for evaluation of players actions for specific states. In particular, the Q-functions introduced in Definition \ref{def:Qfunc} allow us to differentiate between so-called pure strategy and mixed strategy Nash equilibria. 

\begin{definition}[Pure-Strategy Nash Equilibrium]
    \label{def:pureNash}
    A pure-strategy Nash equilibrium has a pair of actions $(a_*^1(s),a_*^2(s))$ upon which neither player will find increased reward by moving unilaterally, 
\begin{align*}
    Q_*^1(s,a_*^1(s),a_*^2(s))&\geq Q_*^1(s,a^1(s),a_*^2(s))\qquad\forall a^1\in A^1\\
    \nonumber
    Q_*^2(s,a_*^1(s),a_*^2(s))&\geq Q_*^2(s,a_*^1(s),a^2(s))\qquad\forall a^2\in A^2\ ,
\end{align*}
where $Q_*^1$ and $Q_*^2$ are the $Q$-functions given in \eqref{eq:Qi}.
\end{definition}
The pure Nash points of Definition \ref{def:pureNash} are a natural way to understand action choices made in a 2-player game, but it is possible that no pure Nash points exist. For each $s\in S$ define the following matrices,
\begin{equation}
    \label{eq:Qmatrices}
    Q_*^i(s) = \Big(Q_*^i(s,a^1,a^2)\Big)_{(a^1,a^2)\in A^1\times A^2}\qquad\hbox{for }i=1,2\ .
\end{equation}
\begin{example}[Actions Spaces with 3 Elements]
Suppose that $A^1= \{1,2,3\}$ and $A^2= \{1,2,3,4\}$. A tuple $\langle Q^1,Q^2,\pi^1,\pi^2\rangle $ has $3\times 4$ matrices
\begin{align*}
Q^i(s) &= 
\begin{pmatrix}
Q^i(s,1,1)&Q^i(s,1,2)&Q^i(s,1,3)&Q^i(s,1,4)\\
Q^i(s,2,1)&Q^i(s,2,2)&Q^i(s,2,3)&Q^i(s,2,4)\\
Q^i(s,3,1)&Q^i(s,3,2)&Q^i(s,3,3)&Q^i(s,3,4)
\end{pmatrix}
\end{align*}
for $i=1,2,$ and for all $s\in S$, and strategy vectors
\[\pi^1(s) = 
\begin{pmatrix}
\pi^1(s,1)\\
\pi^1(s,2)\\
\pi^1(s,3)
\end{pmatrix}~~~\hbox{and}~~~\pi^2(s) = 
\begin{pmatrix}
\pi^2(s,1)\\
\pi^2(s,2)\\
\pi^2(s,3)\\
\pi^2(s,4)
\end{pmatrix}\ ,\]
with $\sum_{\ell=1}^3\pi^1(s,\ell)=1=\sum_{\ell=1}^4\pi^2(s,\ell)$. Given $s_t=s$, the expectation of $Q^i(s,a_t^1,a_t^2)$ is the multiplication of $Q^i(s)$ on the left by $\pi^1(s)$ and from the right by $\pi^2(s)$, 
\[\pi^1(s)^\top Q^i(s)\pi^2(s) = \sum_{\ell=1}^3\sum_{k=1}^4Q^i(s,\ell,k)\pi^1(s,\ell)\pi^2(s,k)\ ,\]
where super-script $\top$ denotes matrix/vector transpose.
\end{example}

For a given $s\in S$, $Q_*^1(s)$ and $Q_*^2(s)$ form a bi-matrix game, for which there may not be a pure Nash strategy, but it was shown by \cite{nash1951non} that a mixed Nash strategy always exists. Before giving the definition for mixed Nash strategies, we first introduce some notation that was also used in \cite{hu1998multiagent}.

\begin{notation}[Distribution Vector Multiplication]
    \label{def:distMultiplication}
    For distribution vectors $\pi^i\in\mathcal P(A^i)$ for $i=1,2$, we let $\pi^1\pi^2Q^i(s)$ denote multiplication of matrix $Q^i(s)$ with vectors $\pi^1(s)$ and $\pi^2(s)$ on the left and from the right, respectively, that is,
    \[\pi^1\pi^2Q^i(s):=\pi^1(s)^\top Q^i(s)\pi^2(s)
    \ ,\]
    where super-script $\top$ denotes matrix/vector transpose.
    \end{notation}

\begin{definition}[Mixed-Strategy Nash Equilibrium]
    \label{def:mixedNash}
    A mixed-strategy Nash equilibrium is a pair of strategies $(\pi_*^1(s),\pi_*^2(s))$ from which neither player will find increased reward by moving unilaterally, 
\begin{align*}
    \pi_*^1\pi_*^2 Q_*^1(s)&\geq \pi^1\pi_*^2 Q_*^1(s)\qquad\forall \pi^1\in\Pi^1\\
    \nonumber
    \pi_*^1\pi_*^2 Q_*^2(s)&\geq \pi_*^1\pi^2 Q_*^2(s)\qquad\forall \pi^2\in\Pi^2\ ,
\end{align*}
where $Q_*^1(s)$ and $Q_*^2(s)$ are matrices like those in \eqref{eq:Qmatrices} with $Q$-functions given by \eqref{eq:Qi}.
\end{definition}
Existence of a mixed Nash strategy for the dynamic game described by equation \eqref{eq:vi} is proven by  \cite{fink1964equilibrium}. The proof uses a Kakhutani fixed point theorem to show that an iteration of an equation similar to \eqref{eq:Qi} will converge to the value function of a Nash equilibrium. In terms of the tuple $\left<Q_*^1,Q_*^2,\pi_*^1,\pi_*^2\right>$, if the mixed-strategy criterion of Definition \ref{def:mixedNash} is satisfied, and if both players adhere to these $Q_*^i$'s and $\pi_*^i$'s for their evaluation of the game and their chosen strategies, then neither player will attempt a unilateral change of strategy, in which case the non-uniqueness of $Q$-functions will not destabilize the equilibrium.

\subsection{Nash Q-Learning}

An algorithm to compute the solution to \eqref{eq:Qi} is the Nash $Q$-learning algorithm of \cite{hu1998multiagent,hu2003nash}, which for the 2-player game is
\begin{align}
    \label{eq:Qi_stoch_iteration}
    Q_{t+1}^i(s_t,a_t^1,a_t^2)&=(1-\alpha_t)Q_{t}^i(s_t,a_t^1,a_t^2)+\alpha_t\Big( r^i(s_t,a_t^1,a_t^2)+\gamma_i\hbox{Nash}Q_t^i(s_{t+1})\Big)\ ,
\end{align}
where $\hbox{Nash}Q_t^i(s) = \pi_t^1\pi_t^2 Q_t^i(s)$ with $(\pi_t^1(s),\pi_t^2(s))$ being a Nash equilibrium for $(Q_t^1(s),Q_t^2(s))$, and with $a_t^i\sim \pi_t^i(s_t,\cdot)$ conditionally independent of $a_t^j$. Equation \eqref{eq:Qi_stoch_iteration} is a reinforcement learning algorithm for obtaining a Nash equilibrium. If for any $s\in S$ the bi-matrix game of $Q_t^1(s)$ and $Q_t^2(s)$ has either a global optimal or a saddle point,\footnote{See \cite{hu1998multiagent} for definition of global and saddle point Nash equilibria.} then $Q_t^i$ converges to $Q_*^i$ of equation \eqref{eq:Qi} as $t\rightarrow \infty$ for learning rate $\alpha_t$ taken to be $\alpha_t=\alpha_t(s_t,a_t^1,a_t^2)$ where
\[\alpha_t(s,a^1,a^2)= 
    \begin{cases}
    c_t\hspace{1cm}\hbox{if }(s,a^1,a^2)=(s_t,a_t^1,a_t^2)\\
    0\hspace{1.1cm}\hbox{otherwise,}
    \end{cases}
\]
with $c_t\in(0,1)$, and with $\sum_{t=0}^\infty\alpha_t(s,a^1,a^2)=\infty$ and $\sum_{t=0}^\infty \big(\alpha_t(s,a^1,a^2)\big)^2<\infty$ a.s. uniformly over $(s,a^1,a^2)\in S\times A^1\times A^2$, see \cite{hu2003nash,jaakkola1993convergence,szepesvari1999unified}. An important step in the implementation of \eqref{eq:Qi_stoch_iteration} is the method for finding a Nash equilibrium in the $t^{th}$ iteration. One possibility is to use the Lemke-Howson algorithm \cite{lemke1964equilibrium} to compute $\hbox{Nash}Q_t^i$ and to sample actions for the next time step.

\section{Learning with Partial-Information}
\label{sec:partialInfo}
The partial-information setting is where Player $i$ does not know the strategy of the opposing player (here forward referred to as Player $j$). For partial information we look for a tuple $\langle\overline{Q}_*^1,\overline{Q}_*^2,\overline{\pi}_*^1,\overline{\pi}_*^2\rangle$ that satisfies the following marginal Q-function equations,
\begin{align}
    \label{eq:marginal_Qi_equation}
    \overline{Q}_*^i(s,a)&=\overline{r}_*^i(s,a) + \gamma_i\sum_{s'\in S}
    \overline{\pi}_*^i(s) \overline{Q}_*^i(s')\overline{p}_*^i(s'|s,a)\\
    \nonumber
    \overline{\pi}_*^i(s)&\in\left\{\pi\in\mathcal P(A^i)\Big|\pi\overline{Q}_*^i(s)\geq \pi' \overline{Q}_*^i(s)\quad\forall \pi'\in\mathcal P(A^i)\right\}\ ,
\end{align}
where $\pi \overline{Q}_*^i(s)=\sum_{a\in A^i}\pi(a) \overline{Q}_*^i(s,a)$ for any $\pi\in\mathcal P(A^i)$, and where we have defined the following marginal quantities,
\begin{align*}
    \overline{r}_*^1(s,a) &= \sum_{a'\in A^2}r^1(s,a,a')\overline{\pi}_*^2(s,a')\\
    \overline{r}_*^2(s,a) &= \sum_{a'\in A^1}r^2(s,a',a)\overline{\pi}_*^1(s,a')\\ \overline{p}_*^i(s'|s,a) &= \sum_{a'\in A^j}\mathbb P(s_{t+1} =s'|S_t =s,a_t^i=a,a_t^j=a')\overline{\pi}_*^j(s,a')\ , 
\end{align*} 
for $i=1,2$ and $j\neq i$. Equation \eqref{eq:marginal_Qi_equation} is solved with only partial information, but a solution provides a full-information Nash strategy:  
\begin{proposition}
    \label{prop:converge_to_full_info_nash}
    Given a tuple $\langle\overline{Q}_*^1,\overline{Q}_*^2,\overline{\pi}_*^1,\overline{\pi}_*^2\rangle$ satisfying \eqref{eq:marginal_Qi_equation}, there exists $Q_*^1$ and $Q_*^2$ that solve \eqref{eq:Qi} with strategy $(\overline{\pi}_*^1,\overline{\pi}_*^2)$, and it follows that $(\overline{\pi}_*^1,\overline{\pi}_*^2)$ is a full-information Nash equilibrium. 
\end{proposition}

\begin{proof}
    Using a contraction principle it is easy to conclude the existence of a unique fixed point $(\widetilde{Q}^1,\widetilde{Q}^2)$ such that
    \begin{align*}
        &\widetilde{Q}^i(s,a^1,a^2) = r^i(s,a^1,a^2)+ \gamma_i \sum_{s'\in S}\overline{\pi}_*^2(s')\overline{\pi}_*^2(s')\widetilde{Q}^i(s')p(s'|s,a^1,a^2)\ ,
    \end{align*}
    from which it follows that $\overline{Q}_*^i(s) = \overline{\pi}_*^j(s)\widetilde{Q}^i(s)$ for $j\neq i$. Then, using the criterion in \eqref{eq:marginal_Qi_equation} for selection of $\overline{\pi}_*^i$, we see that 
    \[\overline{\pi}_*^i\overline{\pi}_*^j\widetilde{Q}^i = \overline{\pi}_*^i\overline{Q}_*^i\geq \pi\overline{Q}_*^i=\pi\overline{\pi}_*^j\widetilde{Q}^i\quad\forall \pi\in\mathcal P(A^i)\ .\]
    Therefore, the tuple $\langle\widetilde{Q}^1,\widetilde{Q}^2,\overline{\pi}_*^1,\overline{\pi}_*^2\rangle$ is equivalent to a solution of \eqref{eq:Qi} for the full-information Q-function, and therefore $(\overline{\pi}_*^1,\overline{\pi}_*^2)$ is a full-information Nash strategy.
\end{proof}
If the marginal $\overline{\pi}_*^i$'s in \eqref{eq:marginal_Qi_equation} can somehow be selected so that they are equal to a joint Nash strategy $(\pi_*^1,\pi_*^2)$ used for computing $Q_*^1$ and $Q_*^2$ in \eqref{eq:Qi}, then there would be no difference in whether or not we work with \eqref{eq:marginal_Qi_equation} or \eqref{eq:Qi}. Indeed, \eqref{eq:marginal_Qi_equation} can be obtained by multiplying both sides of equation \eqref{eq:Qi} by the opposing player's Nash strategy and then summing over their actions. Unfortunately, $(\pi_*^1,\pi_*^2)$ is not known, and so we take the approach of seeking a Nash equilibrium via partially informed learning based on equation \eqref{eq:marginal_Qi_equation}.

The implementation of \eqref{eq:marginal_Qi_equation} will involve a stochastic algorithm where the state process transitions that depend on player actions drawn from estimated partial-information strategies. The following definition will be useful for determining which partial-information strategies are optimal: 
\begin{definition}[Set of Partial-Information Mixed Strategies]
    \label{def:setMixedNash}
    For a given vector $\overline{Q}:S\times A^i\rightarrow \mathbb R$, we denote the set of Player $i$'s partial-information mixed strategies as
    \[\mathcal M^i(\overline{Q}) = \left\{\pi\in\mathcal P(A^i)\Big|\pi \overline{Q}\geq \pi'\overline{Q}\quad\forall \pi'\in\mathcal P(A^i)\right\}\ ,\]
    for $i\in 1,2$.
\end{definition}
A stochastic approximation algorithm for \eqref{eq:marginal_Qi_equation} is
\begin{align}
    \label{eq:marginal_Qi_stoch_iteration}
    \overline{Q}_{t+1}^i(s_t,a_t^i)&=(1-\alpha_t^i)\overline{Q}_{t}^i(s_t,a_t^i)+\alpha_t^i\Big( r^i(s_t,a_t^1,a_t^2)+\gamma_i\overline{\pi}_t^i(s_{t+1})\overline{Q}_t^i(s_{t+1})\Big)\ ,
\end{align}
where player $i$ selects their strategy $\overline{\pi}_t^i(s)\in\mathcal M^i(\overline{Q}_t^i(s))$, and where given $s_t$ the action $a_t^i$ is an independent draw from $\overline{\pi}_t^i(s_t)$,
\[a_t^i\sim \overline{\pi}_t^i(s_t,\cdot)\ ,\]
and with learning rate $\alpha_t^i= \alpha_t^i(s_t,a_t^i)$ where 
\begin{equation}
    \label{eq:alpha_i}
    \alpha_t^i(s,a^i)= 
    \begin{cases}
    c_t\hspace{1cm}\hbox{if }(s,a^i)=(s_t,a_t^i)\\
    0\hspace{1.1cm}\hbox{otherwise,}
    \end{cases}
\end{equation}
where $c_t\in(0,1)$, and where  $\sum_{t=0}^\infty\alpha_t^i(s,a^i)=\infty$ and $\sum_{t=0}^\infty \big(\alpha_t^i(s,a^i)\big)^2<\infty$ a.s. uniformly over $(s,a^i)\in S\times A^i$. The important thing to notice in \eqref{eq:marginal_Qi_stoch_iteration} is that the updated Q-function has no mechanism to force the pair $(\overline{\pi}_t^1,\overline{\pi}_t^2)$ to be a joint Nash strategy; this is different from \eqref{eq:Qi_stoch_iteration} were players take a joint Nash strategy at every $t$.

\begin{remark}[$Q$-Learning for Pure Nash Strategies]
    One might ask why a $Q$-learning algorithm like \eqref{eq:marginal_Qi_stoch_iteration} should be computed with mixed rather than pure Nash strategies. Indeed, the convergence proofs in this paper will remain valid if each iteration searches over pure strategies, and the full-information Nash equilibrium argument of  Proposition \ref{prop:converge_to_full_info_nash} would still be valid for pure Nash strategies. However, if the proof of Proposition \ref{prop:converge_to_full_info_nash} is carried out to include mixed strategies, then the algorithm's limiting pure strategy may not be an equilibrium. In other words, the limit obtained from pure strategies may offer one player (or both players) an improved value function if they deviate unilaterally to a mixed strategy.
\end{remark}

\subsection{Non-Uniqueness and Spurious State Dependence}
\label{sec:uncertainty}

Some interesting things can happen due to the non-uniqueness of solutions to \eqref{eq:marginal_Qi_equation}. In particular, we can have some rather simple games for which there are multiple tuples $\langle Q_*^1,Q_*^2,\pi_*^1,\pi_*^2\rangle$ that define a Nash equilibrium. For example, if $r^1$ and $r^2$ do not depend on $s_t$, then this is just a classical bi-matrix game, and there are mixed Nash strategies $(\pi_*^1,\pi_*^2)$. However, $Q$-learning can lead to tuples $\langle Q_*^1,Q_*^2,\pi_*^1,\pi_*^2\rangle$ where $(\pi_*^1,\pi_*^2)$ is not Nash equilibrium of the bi-matrix game $(r^1,r^2)$, and which also have state dependence. 

If $r^1$ and $r^2$ do not depend on $s_t$, if both players are aware that there is no state dependence, and if both know the exact values of $r^1$ and $r^2$, then a (non-unique) full-information solution is \[Q_*^i(a^1,a^2) = r^i(a^1,a^2) + \gamma_i\pi_*^1\pi_*^2Q_*^i\ ,\] 
where $(\pi_*^1,\pi_*^2)$ is a Nash equilibrium of $(r^1,r^2)$. Here, the non-uniqueness is a direct consequence of the uncertainty issue mentioned in Section \ref{sec:nashEquibriaDefs}. In particular, these two players need to agree which Nash equilibrium to use, or perhaps have it exogenously assigned to them.

Spurious state dependence can arise in this bi-matrix game if we implement $Q$-learning with both players unaware of any Nash equilibria for $(r^1,r^2)$. That is, we consider partial-information equations
\[\overline{Q}_{t+1}^i(s_t,a_t^1) =(1-\alpha_t^i)\overline{Q}_{t}^i(s_t,a_t^1) + \alpha_t^i\Big(r^i(a_t^1,a_t^2) + \gamma_i \pi_t^1\overline{Q}_{t}^i(s_{t+1})\Big)\ , \]
which (as per the main result of this paper) will converge to a tuple $\langle \overline{Q}_*^1,\overline{Q}_*^2,\overline{\pi}_*^1,\overline{\pi}_*^2\rangle$ with strategies $(\overline{\pi}_*^1(s),\overline{\pi}_*^2(s))$ that depend non-trivially on the state, but may not be a Nash equilibrium of $(r^1,r^2)$. The interpretation of this results is that the players have learned a certain way, and their evaluation of the game is based on tuple $\langle \overline{Q}_*^1,\overline{Q}_*^2,\overline{\pi}_*^1,\overline{\pi}_*^2\rangle$ rather that the true bi-matrix game to which they are in fact both playing. The dependence on the state is an artifact of learning, but neither of these players see a benefit to changing their strategy. This is the general interpretation of the $Q$-learning for Nash equilibria: there may be another player with a different evaluation of the game (or perhaps better information) who views the tuple $\langle \overline{Q}_*^1,\overline{Q}_*^2,\overline{\pi}_*^1,\overline{\pi}_*^2\rangle$ as sub-optimal, but for the 2 players whose experience yielded this tuple they do not see any reason to deviate.

\subsection{Filtering and Fictitious Play}
It is possible that players learning under partial information will try to gain an edge by filtering. Suppose that a player's reward only depends on the state and their own action (this is to ensure there is no way to deduce the opponent's actions by observation of rewards). Let $\mathcal F_t^1$ denote the $\sigma$-algebra generated by the observations seen by player 1 up to time $t$. Given $\mathcal F_t^1$ and a strategy $\pi^2(s,a)$ taken by Player 2, Player 1 computes the posterior distribution of the actions taken by player 2,
\[\mathbb P(a_{t'}^2 = a|\mathcal F_t^1\vee\pi^2)=\ddfrac{p(s_{t'+1}|s_{t'},a_{t'}^1,a_{t'}^2 = a)\pi^2(s_{t'},a)}{\sum_{a'\in A^2}p(s_{t'+1}|s_{t'},a_{t'}^1,a_{t'}^2 = a')\pi^2(s_{t'},a')}\quad\forall t'<t\ . \]
However, during learning the opposing player's strategy is not known, and therefore Player 1 will compute the posterior with an estimate. Let $\mathcal F_t^1$ denote the information observed by Player 1 up time time $t$. Player 1 computes an estimator $\widehat{\pi}_t^{2}$ using the Expectation Maximization (EM) algorithm, 
\begin{eqnarray*}
    &\widehat{\pi}_t^{2}=\lim_{\ell\rightarrow\infty}\widehat{\pi}_t^{2,\ell}\ ,
\end{eqnarray*}
where $\widehat{\pi}_t^{2,\ell}$ are computed iteratively as
\begin{eqnarray}
    \label{eq:estimatedPi2}
    &\widehat{\pi}_t^{2,\ell+1}(s,a) = \ddfrac{\sum_{t'=0}^{t-1}\mathbb P(a_{t'}^2 = a|\mathcal F_t^1\vee \widehat{\pi}_t^{2,\ell})\mathbf 1_{s_{t'}=s}}{\sum_{t'=0}^{t-1}\mathbf 1_{s_{t'}=s}}\ ,
\end{eqnarray}
for $\ell=1,2,3\dots$. A standard result is that $\widehat{\pi}_t^{2,\ell}$ are increasing in likelihood and their limit always exists (see \cite{dempster1977maximum}). The following proposition confirms that \eqref{eq:estimatedPi2} are in fact iterations of EM algorithm:
\begin{proposition}
    The parameter estimation sequence in \eqref{eq:estimatedPi2} are iterations of the EM algorithm.
\end{proposition}
\begin{proof}
    Let $a_{0:t}$ and $s_{0:t}$ denote the actions taken and observed state values up to time $t$. For any parameter value $\pi^2$, given $\mathcal F_t^1$ and $\widehat{\pi}_t^{2,\ell}$, iterations of EM algorithm are computed by maximizing the following expected log-likelihood,
    \begin{align*}
        L(\widehat{\pi}_t^{2,\ell},\pi^2)&=\mathbb E\left[\log\mathbb P\left(a_{0:t-1}^2,s_{0:t}|\mathcal F_t^1\vee \pi^2\right)\Big|\mathcal F_t^1\vee \widehat{\pi}_t^{2,\ell}\right]\\
        &=\mathbb E\left[\sum_{t'=0}^{t-1}\log(p(s_{t'+1}|s_{t'},a_{t'}^1,a_{t'}^2)\pi^2(s_{t'},a_{t'}^2)) \Big|\mathcal F_t^1\vee \widehat{\pi}_t^{2,\ell}\right]\ .
    \end{align*}
    We use a Lagrangian to find the parameter that maximizes $L(\widehat\pi_t^{2,\ell} \pi^2)$,
    \[\mathcal L(\pi^2) = L(\widehat{\pi}_t^{2,\ell}, \pi^2) +\sum_{s\in S}\sum_{a\in A^2}\lambda(s)\pi^2(s,a)\ ,\]
    where $\lambda(s)$ are Lagrange multipliers for constraints $\sum_{a\in A^2}\pi^2(s,a)=1$ for each $s\in S$. For a given $s\in S$ and $a\in A^2$, the first-order condition with respect to $\pi^2(s,a)$ is set equal to zero,
    \[\frac{\partial}{\partial \pi^2(s,a)}\mathcal L(\pi^2) = \sum_{t'=0}^{t-1}\left(\frac{\mathbb E[\mathbf 1_{a_{t'}^2=a}|\mathcal F_t^1\vee \widehat{\pi}_t^{2,\ell}]\mathbf 1_{s_{t'}=s}}{\pi^2(s,a)} +\lambda(s)\right)=0\ ,\]
    from which is follows that \eqref{eq:estimatedPi2} is the EM algorithm's updated parameter estimate.    
\end{proof}

For learning, the possible advantage to be gained is to learn an inference-based $Q$-function,
\begin{equation}
    \label{eq:Qhat}
    \widehat Q_{t+1}^1(s_t,a_t^1) = (1-\alpha_t^1)\widehat Q_{t}^1(s_t,a_t^1)+\alpha_t^1\left(r^i(s_t,a_t^1)+\gamma_i\max_{\pi\in\mathcal P(A^1)} \pi\widehat\pi_t^2\widehat Q_{t}^1(s_{t+1})\right)\ .
\end{equation}
This approach to learning could be beneficial if only one of the players does it, but if both players learn this way then there are doubts about convergence to a Nash equilibrium. 

Notice the inferred policy in \eqref{eq:estimatedPi2} is an average of posterior expectations of a full-information policy obtained from fictitious play,
\[\phi_t^{2}(s,a)= \ddfrac{\sum_{t'=0}^t\mathbf 1_{a_{t'}^2 = a}\mathbf 1_{s_{t'}=s}}{\sum_{t'=0}^t\mathbf 1_{s_{t'}=s}}\ .\]
If Player 1 is learning with fictitious play then she assumes Player 2 uses strategy $\phi_t^{2}$ and then computes a best response $\pi_t^{1}$ using an equation like \eqref{eq:Qhat} but with $\phi_t^{2}$ in place of $\widehat{\pi}_t^2$. This fictitious play approach is known to have trouble converging unless Player 2 in fact follows a stationary strategy. For this reason, it is likely that the inference-based learning in \eqref{eq:Qhat} will likely have similar troubles finding a Nash equilibrium. At the time of writing this paper it is unknown to the authors if analysis of \eqref{eq:Qhat} will lead to new conditions for convergence toward a Nash equilibrium, but it is beyond the scope of this paper's Theorem \ref{thm:partial_info_convergence} and the partial-information algorithm based on  \ref{eq:marginal_Qi_equation}.

\subsection{Main Result: Proof of Convergence}
Let $\|\cdot\|$ denote the maximum norm, i.e., $\|\overline{Q}^i(s,a)\| = \max_{a\in A^i}\max_{s\in S}|\overline{Q}^i(s,a)|$. The following lemmas and theorem show that the Q-functions in \eqref{eq:marginal_Qi_stoch_iteration} converge to a solution of \eqref{eq:marginal_Qi_equation}.
    \begin{lemma}
        \label{lemma:partial_info_convergence1}
        Let  $M<\infty$ be a constant such that $\|r^i\|\leq M$ for $i=1,2$. 
        Then for the tuple $\langle\overline{Q}_t^1,\overline{Q}_t^2,\overline{\pi}_t^1,\overline{\pi}_t^2\rangle$ given by \eqref{eq:marginal_Qi_stoch_iteration}, 
        \begin{enumerate}
            \item there are $\overline{Q}_\infty^i$  matrices such that $\|\overline{Q}_t^i-\overline{Q}_\infty^i\|\rightarrow 0$ a.s. as $t\rightarrow \infty$ with $\overline{Q}_\infty^i$ bounded for $i=1,2$, 
            \item there is sub-sequence $t_{\ell}$ and $\overline{\pi}_\infty^i\in\mathcal M^i(\overline{Q}_\infty^i)$ such that  $\|\overline{\pi}_{t_{\ell}}^i-\overline{\pi}_\infty^i\|\rightarrow 0$ a.s. as $\ell\rightarrow \infty$ for $i=1,2$.
        \end{enumerate}
    \end{lemma}
    \begin{proof}
        Part 1) From \eqref{eq:marginal_Qi_stoch_iteration} it can be deduced there is a finite constant depending on $\|\overline{Q}_0^i\|$ and $M$ such that $\|\overline{Q}_{t}^i\|\leq C(\|\overline{Q}_0^i\|,M)$ for all $t\geq 0$, and it follows that 
        \begin{align*}
        &\sum_{t=0}^\infty \|\overline{Q}_{t+1}^i-\overline{Q}_{t}^i\|^2\\
        &\leq\sum_{t=0}^\infty \alpha_t^2\Big|\overline{Q}_{t}^i(s_t,a_t^i)+  r^i(s_t,a_t^1,a_t^2)+\gamma_i\overline{\pi}_t^j(s_{t+1})\overline{Q}_t^i(s_{t+1})\Big|^2\qquad\hbox{for }j\neq i\\
        &\leq 3\Big(2C^2(\|\overline{Q}_0^i\|,M)+M^2\Big)\sum_{t=0}^\infty\alpha_t^2\\
        &<\infty\ ,
        \end{align*}
        where we've used the generic inequality $(a+b+c)^2\leq 3(a^2+b^2+c^2)$. Hence, the tail of the series is convergent, and we conclude there is a finite limit such that $(\overline {Q}_t^1,\overline {Q}_t^2)\rightarrow (\overline {Q}_\infty^1,\overline {Q}_\infty^2)$ a.s. as $t\rightarrow\infty$. . 
        
        Part 2) Because $\mathcal P(A^i)$ is compact, by the Bolzano-Weierstrass theorem we can say that $\overline \pi_t^i$ has a convergent sub-sequence. Then, for any  $\mu(s)\in\mathcal M^i(\overline{Q}_\infty^i(s))$ we have
        \begin{align*}
            0&\leq (\mu(s)-\overline \pi_t^i)\overline Q_\infty(s) \\
            &=(\mu(s)-\overline \pi_t^i)(\overline Q_\infty^i(s)-\overline Q_t^i(s)) +\underbrace{(\mu(s)-\overline \pi_t^i(s))\overline Q_t^i(s)}_{\leq 0} \\
            &\leq (\mu(s)-\overline \pi_t^i(s))(\overline Q_t^i(s)-\overline Q_\infty^i(s))\rightarrow 0\ ,
        \end{align*}
        and so for a convergent sub-sequence $\overline{\pi}_{t_{\ell}}^i$ there is $\overline{\pi}_\infty(s)\in\mathcal M^i(\overline{Q}^i(s))$ such that $\|\overline{\pi}_{t_{\ell}}^i-\overline{\pi}_\infty^i\|\rightarrow 0$ a.s. as $\ell\rightarrow\infty$.  
    \end{proof}
    The limiting tuple $\langle\overline{Q}_\infty^1,\overline{Q}_\infty^2,\overline{\pi}_\infty^1,\overline{\pi}_\infty^2\rangle$ from Lemma \ref{lemma:partial_info_convergence1} cannot be assumed to be a solution to \eqref{eq:marginal_Qi_equation}; instead we must prove it. Convergent sub-sequences of the $\overline{\pi}_t^i$'s from Lemma \ref{lemma:partial_info_convergence1} is an indication that we will be able to show convergence in a similar manner to the convergence results in \cite{hu1998multiagent,szepesvari1999unified} for pseudo-contraction operators \cite{bertsekas1989convergence,feyzmahdavian2014convergence}. 
    
    Consider the $Q$-function $
\widetilde{Q}_{t}^i$ computed using $\overline{\pi}_t^j$,
    \begin{equation}
        \label{eq:fixed_point_t}
        \widetilde{Q}_{t+1}^i(s,a) =\overline{r}_t^i(s,a) + \gamma_i\sum_{s'\in S}\max_{\pi\in\mathcal P(A^i)}\pi\widetilde{Q}_{t}^i(s')\overline{p}_t^i(s'|s,a)\ ,
    \end{equation}
    where 
    \begin{align*}
        \overline{r}_t^1(s,a)&= \sum_{a^2\in A^2} r^1(s,a^1,a^2)\overline{\pi}_t^2(s,a^2)\\
        \overline{r}_t^2(s,a)&= \sum_{a^1\in A^1} r^2(s,a^1,a^2)\overline{\pi}_t^1(s,a^1)\ ,
    \end{align*}
    and $\overline{p}_t^i(s'|s,a)= \sum_{a^j\in A^j}\mathbb P(s_{t+1}=s'|s_t=s,a_t^i=a,a_t^j=a^j)\overline{\pi}_t^j(s_t,a^j)$.

\begin{lemma}
    \label{lemma:partial_info_convergence2}
    Let $M<\infty$ be a constant such that $\|r^i\|\leq M$ for $i=1,2$. 
    Then $\|\overline{Q}_t^i-\widetilde{Q}_{t}^i\|\rightarrow  0$ a.s. as $t\rightarrow \infty$, for $i=1,2$.
\end{lemma}
\begin{proof}
    Without loss of generality, we can prove the limit for Player 1. Define
    \[\Delta_t(s,a^1) := \overline{Q}_{t}^1(s,a^1)-\widetilde{Q}_{t}^1(s,a^1)\ ,\]
    for all $(s,a^1)\in S\times A^1$. From \eqref{eq:marginal_Qi_stoch_iteration} and \eqref{eq:fixed_point_t} it follows that
    \begin{align*}
        &\Delta_{t+1}(s,a^1) = (1-\alpha_t^1(s,a^1))\Delta_{t}(s,a^1)+\alpha_t^1(s,a^1) w_{t+1}\ ,
    \end{align*}
    with $\alpha_t^1(s,a^1)$ as defined in \eqref{eq:alpha_i} and where 
    \begin{align*}
        &w_{t+1} = r^1(s_t,a_t^1,a_t^2)-\sum_{a^2\in A^2}r^1(s_t,a_t^1,a^2)\overline{\pi}_t^2(s_t,a^2)\\
        &~~~~+\gamma_1\left(\max_{\pi\in \mathcal P(A^1)}\pi \overline{Q}_t^1(s_{t+1}) -\sum_{a^2\in A^2}\sum_{s'\in S}\max_{\pi\in \mathcal P(A^1)}\pi\widetilde{Q}_{t}^1(s')p(s'|s_t,a_t^1,a^2)\overline{\pi}_t^2(s_t,a^2)\right)\ .
    \end{align*}
    Let $\mathcal G_t$ denote the $\sigma$-algebra generated by $(\Delta_u(s,a^1))_{u\leq t}$, $(\alpha_u(s,a^1))_{u\leq t}$, $(w_u)_{u\leq t}$, $(s_u,a_u^1)_{u\leq t}$, and $(\overline\pi_u^2)_{u\leq t}$ (note that $a_{t}^2$ is not observable under $\mathcal G_t$, but $\widetilde{Q}_{t}^1$ is observed). Because $a_t^2$ is an independent draw from $\overline{\pi}_t^2(s_t,\cdot)$, it follows that 
    \begin{eqnarray*}
    \Big|\mathbb E[w_{t+1}|\mathcal G_t]\Big|&\leq& \gamma_1\|\Delta_t\|\\
    \hbox{var}(w_{t+1}|\mathcal G_t)&=&\hbox{var}(r^1(s_t,a_t^1,a_t^2)+ \gamma_1\max_{\pi\in \mathcal P(A^1)}\pi \overline{Q}_t^1(s_{t+1})|\mathcal G_t)\\
    &=&\hbox{var}(r^1(s_t,a_t^1,a_t^2)+ \gamma_1(\max_{\pi\in \mathcal P(A^1)}\pi \overline{Q}_t^1(s_{t+1})-\max_{\pi\in \mathcal P(A^1)}\pi \widetilde{Q}_{t}^1(s_{t+1}))\\
    &&\hspace{2cm}+\gamma_1\max_{\pi\in \mathcal P(A^1)}\pi \widetilde{Q}_{t}^1(s_{t+1})|\mathcal G_t)\\
    &\leq &3M^2+3\gamma_1^2\|\Delta_t\|^2+3\gamma_1^2\mathbb E\left[\|\widetilde{Q}_{t}^1\|^2\Big|\mathcal G_t\right]\\
    &\leq &3\left(M^2 + \gamma_1^2\left(\frac{M}{1-\gamma_1}+\|\widetilde{Q}_0^i\|\right)^2\right)+3\gamma_1^2\|\Delta_t\|^2
    \end{eqnarray*}
    where we've used the generic inequality $(a+b+c)^2\leq 3(a^2+b^2+c^2)$, the generic inequality $|\max_xf(x)-\max_xg(x)|\leq \max_x|f(x)-g(x)|$, and we have used the bound $\|\widetilde{Q}_{t}^1\|\leq \frac{M}{1-\gamma_1}+\|\widetilde Q_0^i\|$ (deduced from \eqref{eq:fixed_point_t}). From Theorem 1 in \cite{jaakkola1993convergence} it follows that $\Delta_t=\widetilde{Q}_{t}^1-\widetilde{Q}_{t}^1$ converges to zero a.s. as $t\rightarrow \infty$. 
\end{proof}
\begin{theorem}
    \label{thm:partial_info_convergence}
    Let $M<\infty$ be a constant such that $\|r^i\|\leq M$ for $i=1,2$. 
    The tuples $\langle\overline{Q}_t^1,\overline{Q}_t^2,\overline{\pi}_t^1,\overline{\pi}_t^2\rangle$ given by \eqref{eq:marginal_Qi_stoch_iteration} produce a limit $\langle\overline{Q}_\infty^1,\overline{Q}_\infty^2,\overline{\pi}_\infty^1,\overline{\pi}_\infty^2\rangle$ that is a solution to \eqref{eq:marginal_Qi_equation}, where $(\overline{Q}_t^1,\overline{Q}_t^2)\rightarrow (\overline{Q}_\infty^1,\overline{Q}_\infty^2)$ a.s. as $t\rightarrow \infty$ and there is a sub-sequence $t_\ell$ such that $(\overline{\pi}_{t_\ell}^1,\overline{\pi}_{t_\ell}^2)\rightarrow(\overline{\pi}_\infty^1,\overline{\pi}_\infty^2)$ a.s. as $\ell\rightarrow\infty$.
\end{theorem}
\begin{proof}
    From Lemma \ref{lemma:partial_info_convergence1} we know there is a limit for the tuple $\langle\overline{Q}_t^1,\overline{Q}_t^2,\overline{\pi}_t^1,\overline{\pi}_t^2\rangle$, where $(\overline{Q}_t^1,\overline{Q}_t^2)\rightarrow (\overline{Q}_\infty^1,\overline{Q}_\infty^2)$ a.s. as $t\rightarrow\infty$, and where there is a sub-sequence $t_\ell$ such that $(\overline{\pi}_{t_\ell}^1,\overline{\pi}_{t_\ell}^2)\rightarrow (\overline{\pi}_\infty^1,\overline{\pi}_\infty^2)$ a.s. as $\ell\rightarrow \infty$, with $\overline{\pi}_\infty^i\in\mathcal M^i(\overline{Q}_\infty^i)$ for $i=1,2$. The theorem will be proved if we can show that $\overline{Q}_\infty^i$ and $\overline{\pi}_\infty^j$ for $j\neq i$ solve \eqref{eq:marginal_Qi_equation}. Observe the following:
    \begin{align}
    \nonumber
    &\left| \overline{Q}_\infty^1(s,a^1) -\overline{r}_\infty^1(s,a^1)-\gamma_1 \sum_{a^2\in A^2}\sum_{s'\in S}\max_{\pi\in\mathcal P(A^1)}\pi \overline{Q}_\infty^1(s')p(s'|s,a^1,a^2)\overline{\pi}_\infty^2(s,a^2)\right|\\
    \nonumber
    &\leq\left| \overline{Q}_\infty^1(s,a^1)-\widetilde{Q}_{t+1}^1(s,a^1) \right|\\
    \nonumber
    & + \left|\overline{r}_\infty^1(s,a^1)-\overline{r}_t^1(s,a^1)\right|\\
    \nonumber
    &+\gamma_1 \sum_{a^2\in A^2}\sum_{s'\in S}\max_{\pi\in\mathcal P(A^1)} \left|\pi\overline{Q}_\infty^1(s')\overline{\pi}_\infty^2(s,a^2)-\pi\widetilde{Q}_t^1(s')\overline{\pi}_t^2(s,a^2)\right|p(s'|s,a^1,a^2)\\
    \nonumber
    &\leq \| \overline{Q}_\infty^1-\widetilde{Q}_{t+1}^1\|+\gamma_1\| \overline{Q}_\infty^1-\widetilde{Q}_{t}^1\| + (M+C(\|\overline{Q}_0^i\|,M))\|\overline{\pi}_\infty^2-\overline{\pi}_t^2\|\\
    \nonumber
    &\leq \| \overline{Q}_\infty^1-\overline{Q}_{t+1}^1\|+\| \overline{Q}_{t+1}^1-\widetilde{Q}_{t+1}^1\|\\
    \nonumber
    &~~~~+\gamma_1\left(\| \overline{Q}_\infty^1-\overline{Q}_{t}^1\|+\| \overline{Q}_{t}^1-\widetilde{Q}_{t}^1\|\right)\\
    \label{eq:proofCalc}
    &~~~~+
    (M+C(\|\overline{Q}_0^i\|,M))\|\overline{\pi}_\infty^2-\overline{\pi}_t^2\|\ ,
    \end{align}
    where $C(\|\overline{Q}_0^i\|,M)$ is the bound from Lemma \ref{lemma:partial_info_convergence1} such that $\|\overline{Q}_\infty^1\|\leq C(\|\overline{Q}_0^i\|,M)<\infty$. Now,  $\|\overline{\pi}_\infty^2 - \overline{\pi}_{t_{\ell}}^2\|$ converges to zero a.s. as $\ell\rightarrow\infty$, and therefore $\|\overline{Q}_\infty^1 - \widetilde{Q}_{t_\ell}^1\|\rightarrow 0$ a.s. as $\ell\rightarrow \infty$ by Lemma \ref{lemma:partial_info_convergence2}, and so all the terms in the bottom line of \eqref{eq:proofCalc} go to zero. Hence,
    \[\left|\overline{Q}_\infty^1(s,a^1) -\overline{r}_\infty^1(s,a^1)-\gamma_1 \sum_{a^2\in A^2}\sum_{s'\in S}\max_{\pi\in\mathcal P(A^1)}\pi \overline{Q}_\infty^1(s')p(s'|s,a^1,a^2)\overline{\pi}_\infty^2(s,a^2)\right|=0\ ,\]
     and we conclude that $\langle\overline{Q}_\infty^1,\overline{Q}_\infty^1,\overline{\pi}_\infty^1,\overline{\pi}_\infty^2\rangle$ is a solution of \eqref{eq:marginal_Qi_equation}.
\end{proof}

\section{Simulated Examples on General 2-Player Games}
\label{sec:simulations}
In this section we present output for three examples: a randomly generated bi-matrix game, the 2-player Gridworld game, and LeDuc Hold'em poker. Gridworld and LeDuc Hold'em are large enough simulations that we need to implement a deep neural network algorithm to reduce runtime. The results in this section empirically show convergence to Nash policies.

\subsection{Simulated Random Bi-Matrix Game}
As a first demonstration we implement partial-information Q-learning on a randomly generated bi-matrix game. Let the action spaces be $A^1 = \{0,1,2,\dots,d_1-1\}$ and $A^2 =\{0,1,2,\dots,d_2-1\}$ with $d_1 =5$, $d_2 =7$, and let the state space be $S=\{0,1,2,3,\dots,d_s-1\}$ with $d_s=10$. Take the rewards to be $r^i\in \mathbb R^{d_1\times d_2}$ with
	\begin{align*}
		r^1(s,a^1,a^2)&\sim U(0,1)\\
		r^2(s,a^1,a^2) &\sim h \times r^1(s,a^1,a^2) + (1-h)\times U(0,1)\ ,
	\end{align*}
	with $h=.8$, $\gamma_1 = .9$ and $\gamma_2 = .8$, and where $U(0,1)$ denotes an independent draw from a uniform distribution on $[0,1]$. We also take the transition probabilities for $s_t$ to be proportional to uniform draws,
	\[p(s_{t+1}=s'|s_t = s,a_t^1 = a,a_t^2=a')\propto U(0,1)\ ,\]
	for all $s,s'\in S$ and $a\in A^1 $ and $a'\in A^2$. We run Algorithm \ref{alg:Qlearning} for 4,000 iterations with $\alpha_t = \frac{1}{1+\lfloor \tfrac{t}{250}\rfloor}$. After learning strategies $\overline\pi_*^1$ and $\overline\pi_*^2$ we compute the full-information $Q$ functions as per Proposition \ref{prop:converge_to_full_info_nash},
	\begin{align*}
		Q_*^i(s_t,a_t^1,a_t^2) &= r^i(s_t,a_t^1,a_t^2)+\gamma_i\mathbb E_{s_t,a_t^1,a_t^2}\overline\pi_*^1\overline\pi_*^2 Q_*^i(s_{t+1})\ .
	\end{align*}
	Figures \ref{fig:biMatrixSimQ1andQ2}, \ref{fig:biMatrixSimFullConvergence} and \ref{fig:biMatrixSimQ1andQ2_all} show the results of this randomly generated game. Partial information works well but also requires solving of 4 linear programs at each iteration, and so runtime is not much faster than full-information using Lemke-Howson.

\begin{algorithm}
\caption{Nash Q learning for 2-Player Partial-Information Game}
\label{alg:Qlearning}
\begin{algorithmic}
\State Initialize state $s_0$ and vectors $\overline Q^1,\overline{Q}^2$;
\For{$t=0,1,2...,$ max\_iter}
    \For{$i=1,2$}
        \State \hbox{compute } $\overline\pi^i(s_t)\in\mathcal M(\overline Q^i(s_t))$;
        \State \hbox{draw action }$a^i\sim \overline\pi^i(s_{t})$;
    \EndFor
    \State \hbox{draw state }$s_{t+1}\sim p(\cdot|s_{t},a^1,a^2)$;
    \For{$i=1,2$}
        \State \hbox{compute } $\overline\pi^i(s_{t+1})\in\mathcal M(\overline Q^i(s_{t+1}))$;
    	\State $\overline{Q}^i(s_{t},a^i) \gets (1-\alpha_t)\overline{Q}^i(s_t,a^i) +\alpha_t\left( r^i(s_{t},a^1,a^2)+\gamma_i\overline{\pi}^i(s_{t+1})\overline{Q}^i(s_{t+1})\right)$;
    \EndFor
\EndFor
\end{algorithmic}
\end{algorithm}

\begin{figure}[h!]
    \centering
    \includegraphics[width=.75\textwidth]{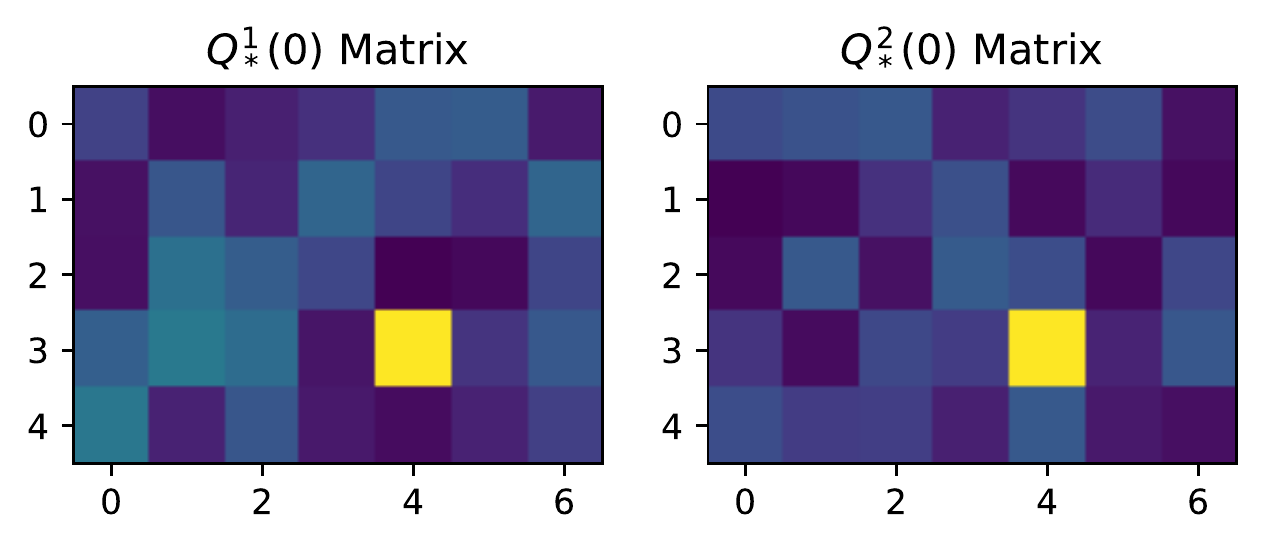}
    \caption{For $s = 0$ in the simulated bi-matrix game, the full-information Q matrices computed using the partial-information strategies, $Q_*^i(s_t,a_t^1,a_t^2) = r^i(s_t,a_t^1,a_t^2)+\gamma_i\mathbb E_{s_t,a_t^1,a_t^2}\overline\pi_*^1\overline\pi_*^2 Q_*^i(s_{t+1})$, shown here for state $s_t=0$ wherein the Nash equilibrium is $\overline\pi^1(0) = (1,0,0,0,0)$ and $\overline\pi^2(0) = (0,0,0,0,1,0,0)$ form a Nash equilibrium.}
    \label{fig:biMatrixSimQ1andQ2}
\end{figure}

\begin{figure}[h!]
    \centering
    \includegraphics[width=.75\textwidth]{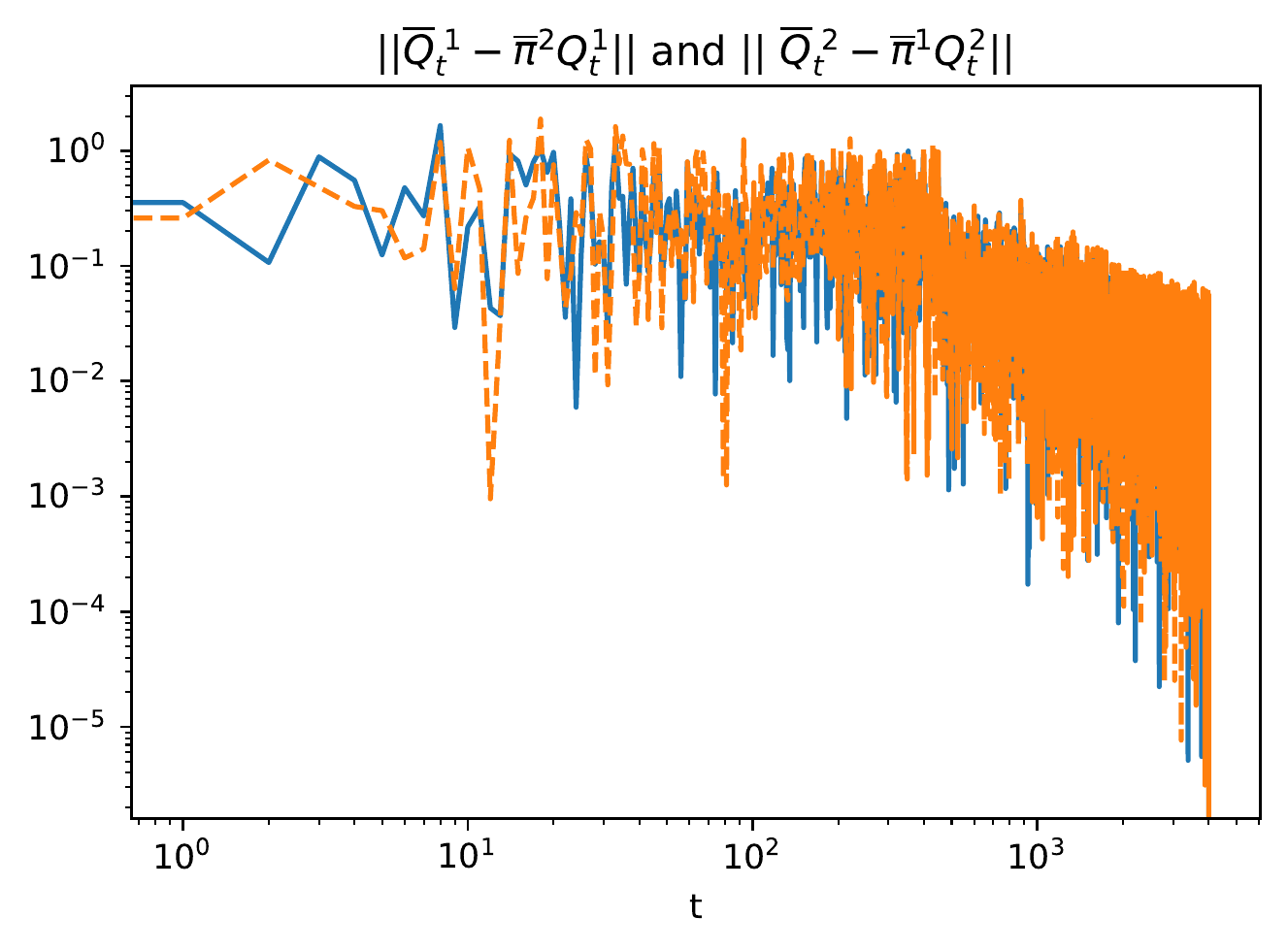}
    \caption{This plot shows convergence in the simulated bi-matrix for the partial-information $Q$-function relative to the full-information $Q$-function multiplied by $\overline{\pi}_t$.}
    \label{fig:biMatrixSimFullConvergence}
\end{figure}

\begin{figure}[t]
    \centering
    \minipage{\textwidth}
     \includegraphics[width=\textwidth]{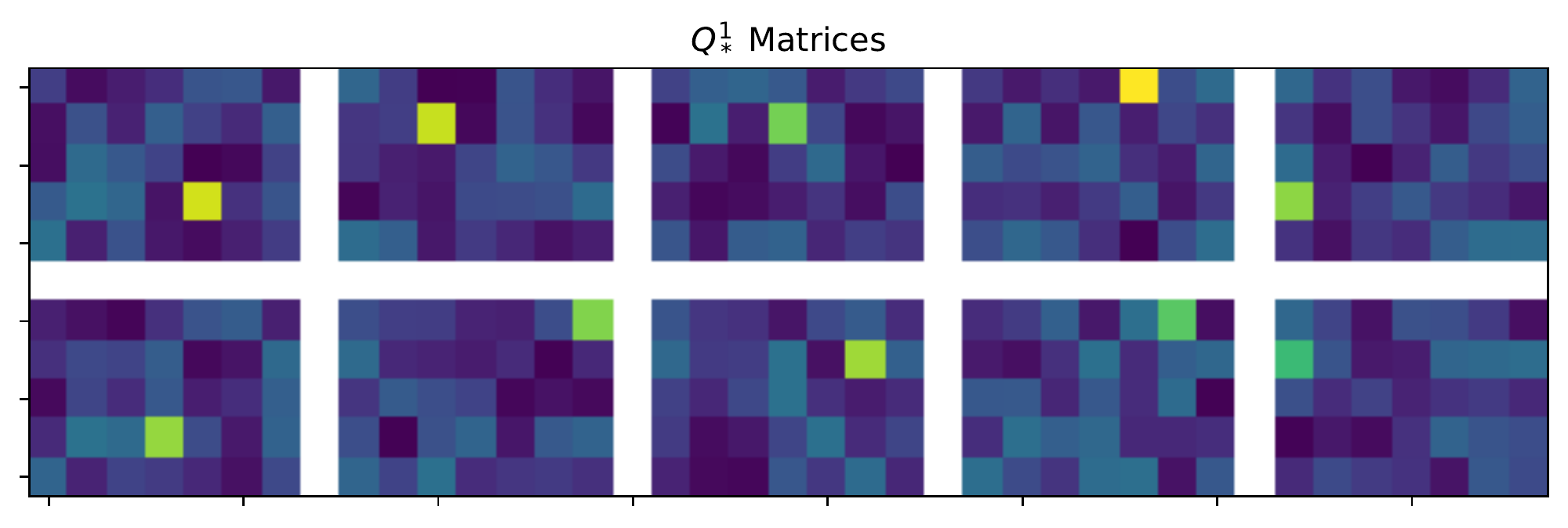}
    \endminipage\hfill
    \minipage{\textwidth}
    \includegraphics[width=\textwidth]{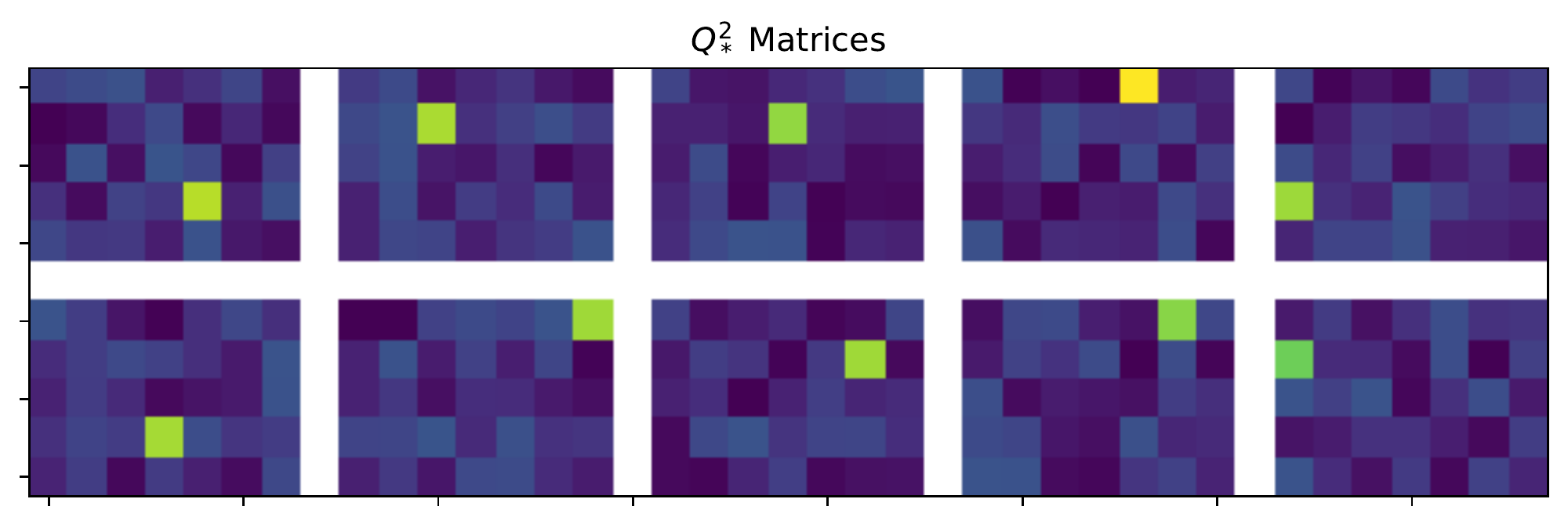}
    \endminipage
    \caption{For the simulated bi-matrix game, these are the full-information $Q_*^i(s_t,a_t^1,a_t^2) = r^i(s_t,a_t^1,a_t^2)+\gamma_i\mathbb E_{s_t,a_t^1,a_t^2}\overline\pi_*^1\overline\pi_*^2 Q_*^i(s_{t+1})$ where $(\overline\pi_*^1,\overline\pi_*^2)$ is the Nash equilibrium from the  partial-information algorithm. There are 10 possible values for $s_t$, and so we see 10 $Q$-functions for each player.}
    \label{fig:biMatrixSimQ1andQ2_all}
\end{figure}

\subsection{Examples from the Python RL Library}
For increasingly larger-scale games intensifies the need to use Deep Neural Networks: Deep layers help us learn more aspects of the game during the experience replay/exploration phase. The neural network approximation of the $Q$ functions is
\[Q_*^i(s,a)\approx  Q(s,a|\theta_*^i)\ ,\]
where $\theta_*^i$ are the optimal network parameters for Player $i$. As our Algorithm \ref{alg:DQN_partial} dictates, and as shown in Figure \ref{fig:General DQN architecture}, every player in the games we simulate will be associated with a neural network consisting of densely connected layers. Each neural network has an input layer for which the number of nodes is fixed and corresponds to the game's state-space shape, some hidden optimizing layers, and an output layer that also has a fixed number of nodes corresponding to the number of actions in an environment. The players progressively learn the optimal policies and the corresponding optimal value function through experience replay \cite{DBLP:journals/corr/abs-2003-02372}; in Algorithm \ref{alg:DQN_partial} there is a `max\_memory' parameter that controls the length of the replay buffer. The learning is optimized by alternating between exploration and exploitation \cite{https://doi.org/10.48550/arxiv.2109.06668}. The novelty in the deep learning models introduced in this section is the use of our partial-information Q-learning equations instead of the full-information updates. The purpose of modeling large-scale games is to test the effectiveness of our method for more complex models. We also want to compare our method's solution to results found when we solve the same games for fully informed players.

The number of hidden layers, the number of its nodes, and the choice of activation functions are hyper-parameters of the architecture. For simulations and neural network implementation, we use the Python library Keras \cite{ketkar2017introduction}.

\begin{figure}[h!]
    \centering
    \includegraphics[width=.85\textwidth]{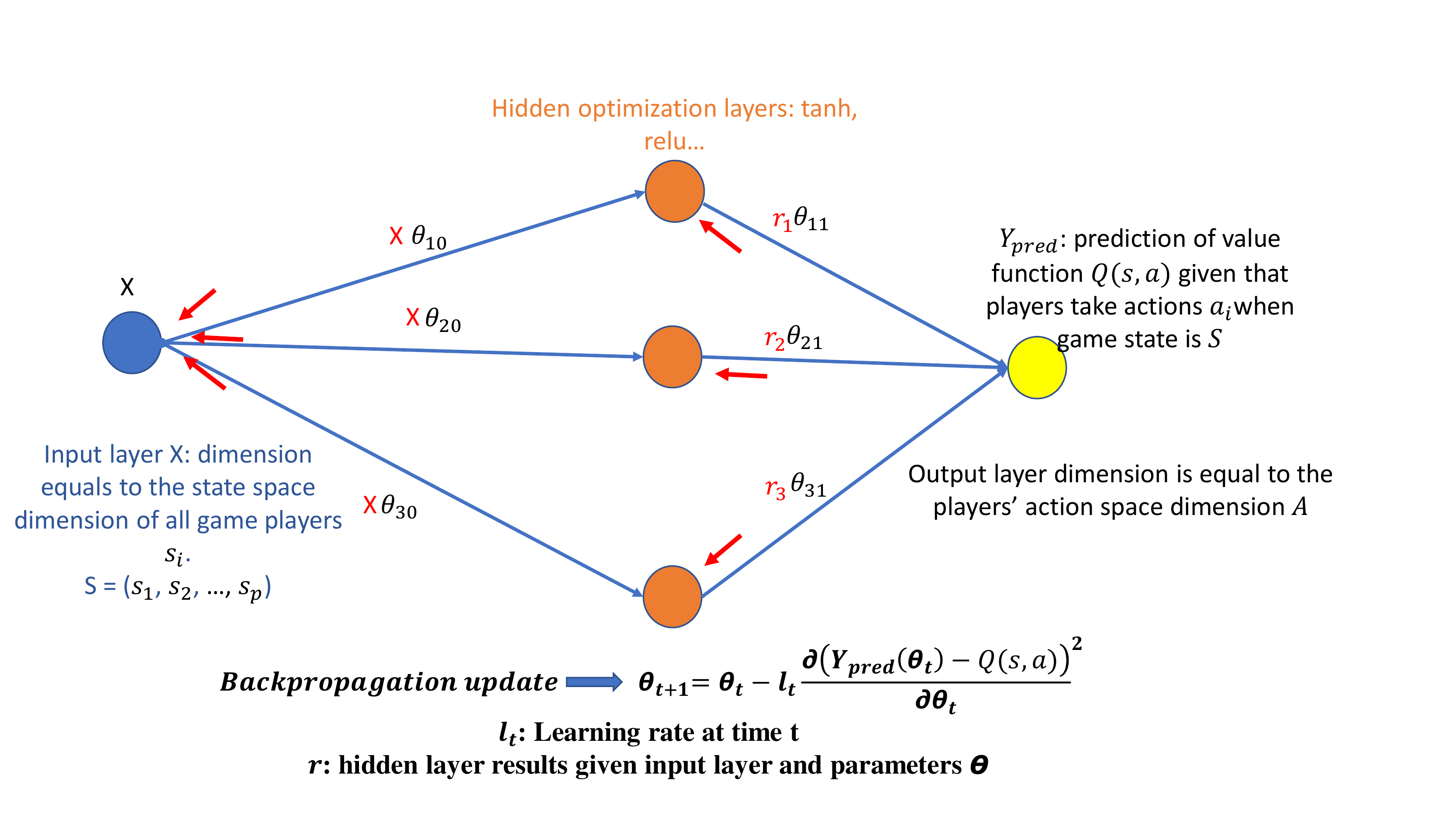}
    \caption{This figure presents the general architecture used to solve for optimal players strategies when faced with larger scale games. The choice of activation functions and number of hidden layers are hyper parameters that should be considered carefully for each game.}
    \label{fig:General DQN architecture}
\end{figure}

The Gridworld and LeDuc Hold'em games are examples of episodic learning, whereas the example games from earlier in this paper were learned over an infinite horizon. Episodes, as well as the new learning features (i.e., experience replay) that we introduce, are outside the assumptions of Theorem \ref{thm:partial_info_convergence}, but partial information is still an integral part of Algorithm \ref{alg:DQN_partial}. Indeed, theoretical analysis of Gridworld and LeDuc Hold'em learning via Algorithm \ref{alg:DQN_partial} will be part of future follow-up work from paper.

\subsubsection{Gridworld 2-Player Game}
The Gridworld game we define is presented in Figure \ref{fig:gridworld}. Our game has deterministic moves: Two agents start from respective lower corners, trying to reach their target cells in the top row. Each agent can only move one cell at a time, and in four possible directions: Left, Right, Up, Down. If two agents attempt to move into the same cell or move outside the map, they are bounced back to their previous cells and penalized with a negative reward. The game ends as soon as at least one agent reaches its goal. Reaching the target cell earns a positive reward. In case both agents
reach their target cells at the same time, both are rewarded with positive payoffs. The objective of an agent in this game is to reach target cells as soon as possible. Since both agents can win if they reach target at the same time. A natural Nash equilibrium is a sequence of moves that result in both agents attempting to win with a minimum number of steps. We assume that agents do not know the locations of their targets at the beginning of the learning period. Furthermore, players do not know their own and the other the player's reward functions. They observe only their own current location and the other player's current location before making a decision. 

The Action space is defined as follows:
\begin{align}
        A^1=A^2=\{\hbox{Left, Right, Up, Down}\},
\end{align}
We also define the state space as the the tuple of both players' location: $S=(L^1,L^2)=\{(0,1),(1,2)...(79,80)\}$ where we define $L^1$ and $L^2$ to be the location of Player 1 and Player 2 on the grid. Finally, we define the reward functions for each player $i=1,2$:
\begin{equation*}
r^i(S,a^1,a^2) = \begin{cases}
10 &\text{if $L^i=Target^i$ }\\
-0.5 &\text{if $L^1=L^2$ or move outside map}\\
0    & \text{else.}
\end{cases}
\end{equation*}
Figure \ref{fig:gridworld_conv} and Table \ref{tab:opt_grid} show the learning process results of our Q-agents. We see that our algorithm is able to converge to the optimal Nash equilibrium after running for approximately 6000 episodes. The theoretical optimal Nash equilibrium occurs when both players exit at the same time for a cumulative reward of 20 points. Highlighting the importance of access to complete state-space information, we solve the problem for blind players (i.e., players only have information about their own location) and compare results. The same algorithm fails to learn the optimal Nash equilibrium when players are blind, which demonstrates the importance of access to state-space information.
\begin{table}[ht]
    \centering
        \begin{tabular}{ p{10em} p{10em} p{14em} }
        \hline
        \textbf{Player type} & \textbf{Convergence} & \textbf{Reaching Nash point} \\
        \hline
        \\
        Partially informed & Converges &  100\% of the time \\
        \\
        \hline
        \\
        Blind  & No convergence & $\approx 15\%$ of the time \\
        \\
        \hline
        \\
    \end{tabular}
    \caption {2-player Gridworld convergence results}
    \label{tab:opt_grid}
\end{table}

\begin{figure}[H]
\minipage{0.48\textwidth}
  \includegraphics[width=\linewidth]{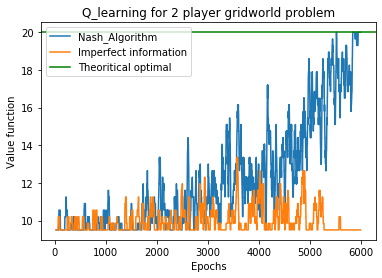}
  \caption{DQN vs Blind Agent Nash learning}\label{fig:gridworld_conv}
\endminipage\hfill
\minipage{0.48\textwidth}
  \includegraphics[width=\linewidth]{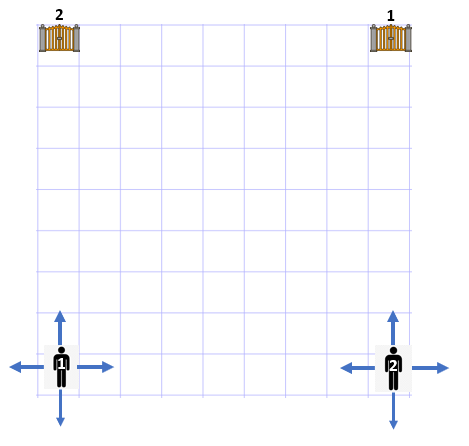}
  \caption{The Gridworld game}\label{fig:gridworld}
\endminipage
\end{figure}

\begin{algorithm}
\caption{DQN Nash learning for 2-Player Partial-Information Game}
\label{alg:DQN_partial}
\begin{algorithmic}
\State Initialize $\theta_0^1,\theta_0^2$;
\State Initialize experience replay memory: $M_Q^i = \emptyset$;
\State Initialize hyper parameters for learning rate decay: $decay$;
\State Initialize estimated number of epochs : $T$;
\State Initialize exploration-exploitation rate: $\epsilon_1 = 1$;
\For{$t=1,2...$ till convergence}
    \State \# Initialize game and observe initial state $s$;
    
    \While{$s'$ not terminal}
    \State $U\sim \hbox{uniform}(0,1)$;
    \For{$i=1,2$}
    \If {$U\leq\epsilon_t$}
    \State $a^i\sim\hbox{uniform}(A^i)$;
    \Else
    \State $a^i=\argmax_{a\in A^i}Q^i(s,\cdot|\theta_{t-1})$; 
    \EndIf
    \EndFor
    \State $s'\sim p(\cdot|s,a^1,a^2)$;
    \For{$i=1,2$}
    \State Set $\widehat{Q}^i(s,a^i) \gets r^i(s,a^1,a^2)+\gamma_i\max_{\pi\in \mathcal P(A^i)}\pi Q^i(s'|\theta_{t-1})$;
    \State Set $M_Q^i\gets M_Q^i\cup \{(s,a^i,r^i,s',\widehat{Q}^i(s,a^i)\}$;
    \If{$t>$max\_memory}
    \State Remove oldest element in $M_Q^i$
    \EndIf
    \EndFor
    
    \EndWhile
    
    \For{$i=1,2$}
   
    \State Set Loss($\theta^i$)= $\mathbb E_{(s,a^i,r^i,s') \sim M_Q^i}[(\widehat{Q}^i(s,a^i)-Q^i(s,a^i|\theta^i))^2]$;
    \State Update: $\theta_t^i=\theta_{t-1}^i-l_t^i\nabla_\theta \hbox{Loss}(\theta_{t-1}^i)$;   
    \EndFor
    \State $\epsilon_{t+1} \gets \epsilon_{t}- \frac{1}{t}$;
    \State $l_{t+1}\gets\frac{l_{t}}{1+ decay*T} $;

\EndFor
\end{algorithmic}
\end{algorithm}

\subsubsection{Leduc Hold'em game}
In this section we test our algorithm on the Leduc Hold'em poker 2-player game. Leduc Hold'em is a toy poker game that has been widely studied \cite{southey2012bayes} and is simulated using the python library Rlcard \cite{zha2019rlcard}. It is usually played with a deck of six cards, comprising two suits (two kings, two queens, and two jacks). Each player is dealt one private card, followed by a betting round. Then, another card is dealt face-up followed by another betting round. Finally, the players reveal their private cards. If one player's private card is the same rank as the board card, he or she wins the game; otherwise the player whose private card has the higher rank wins.

For our test with Leduc Hold'em poker game we define three scenarios. In the first scenario we model a Neural Fictitious Self Player \cite{zhang2021monte} competing against a random-policy player. Figure \ref{fig:Fictitious vs Random} shows the learning process of our NFSP player who is able to earn more rewards, taking advantage of the ``irrationality" of his opponent. In the second scenario, 2 NFSP players are competing against each other. Results shown in Figure \ref{fig:Fictitious vs Fictitious} demonstrate that our player earns lower rewards when facing another ``rational" player. In the third scenario, strategies made by both players follow our DQN algorithm. Results in Figure \ref{fig:DQN vs DQN} show that the learning process is more volatile due to the imperfect information setting. However, the obtained Nash equilibrium is comparable to the one found in scenario 2. This shows that our algorithm is a viable and reliable alternative to other algorithms that require a significantly higher computational effort.

\begin{figure}[H]
    \centering
    \includegraphics[width=6.5cm,height=5cm]{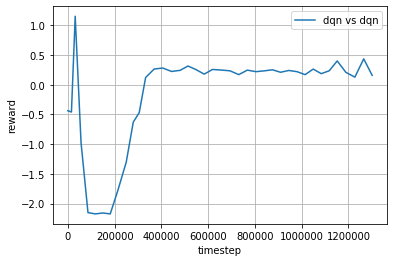}
    \caption{DQN vs DQN}
    \label{fig:DQN vs DQN}
\end{figure}

\begin{figure}[H]
\minipage{0.45\textwidth}
  \includegraphics[width=\linewidth]{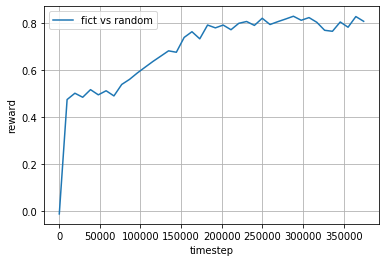}
  \caption{NFSP vs Random}\label{fig:Fictitious vs Random}
\endminipage\hfill
\minipage{0.45\textwidth}
  \includegraphics[width=\linewidth]{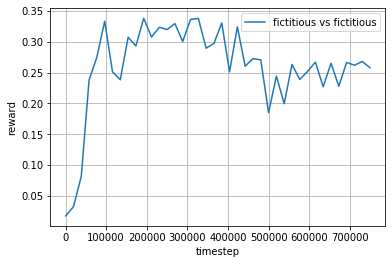}
  \caption{NFSP vs NFSP}\label{fig:Fictitious vs Fictitious}
\endminipage
\end{figure}

\section{Conclusion}
We have presented an analysis of partial-information Q-learning, which is similar to the Q-learning algorithm of \cite{hu1998multiagent}. One of the benefits of this approach is that is bypasses the Lemke-Howson algorithm when updating the Q-functions. We have implemented our algorithm on the Gridworld game and LeDuc poker using deep neural network approximation, and have shown results indicating that the algorithm converges to a Nash equilibrium that is comparable to those obtained from full-information Q-learning and neural fictitious play. 

The follow-up work for this paper is several fold. One direction is an investigation of how this algorithm scales, which we would do with new examples of increasing dimensionality (i.e., large actions spaces, large state space, and more than 2 players). Another direction would be to research some of the nuances of experience replay and to develop possibly theory for understanding of convergence when it is used.


\bibliography{nashAlgoRefs.bib}

\begin{thebibliography}{10}

\bibitem{bertsekas1989convergence}
Dimitri~P Bertsekas and John~N Tsitsiklis.
\newblock Convergence rate and termination of asynchronous iterative
  algorithms.
\newblock In {\em Proceedings of the 3rd International Conference on
  Supercomputing}, pages 461--470, 1989.

\bibitem{brown1951iterative}
George~W Brown.
\newblock Iterative solution of games by fictitious play.
\newblock {\em Activity analysis of production and allocation}, 13(1):374--376,
  1951.

\bibitem{dempster1977maximum}
Arthur~P Dempster, Nan~M Laird, and Donald~B Rubin.
\newblock Maximum likelihood from incomplete data via the em algorithm.
\newblock {\em Journal of the Royal Statistical Society: Series B
  (Methodological)}, 39(1):1--22, 1977.

\bibitem{feyzmahdavian2014convergence}
Hamid~Reza Feyzmahdavian and Mikael Johansson.
\newblock On the convergence rates of asynchronous iterations.
\newblock In {\em 53rd IEEE Conference on Decision and Control}, pages
  153--159. IEEE, 2014.

\bibitem{fink1964equilibrium}
Arlington~M Fink.
\newblock Equilibrium in a stochastic $ n $-person game.
\newblock {\em Journal of science of the hiroshima university, series ai
  (mathematics)}, 28(1):89--93, 1964.

\bibitem{fudenberg1993learning}
Drew Fudenberg and David~M Kreps.
\newblock Learning mixed equilibria.
\newblock {\em Games and economic behavior}, 5(3):320--367, 1993.

\bibitem{heinrich2016deep}
Johannes Heinrich and David Silver.
\newblock Deep reinforcement learning from self-play in imperfect-information
  games.
\newblock {\em arXiv preprint arXiv:1603.01121}, 2016.

\bibitem{hu2003nash}
Junling Hu and Michael~P Wellman.
\newblock Nash {Q}-learning for general-sum stochastic games.
\newblock {\em Journal of machine learning research}, 4(Nov):1039--1069, 2003.

\bibitem{hu1998multiagent}
Junling Hu, Michael~P Wellman, et~al.
\newblock Multiagent reinforcement learning: theoretical framework and an
  algorithm.
\newblock In {\em ICML}, volume~98, pages 242--250. Citeseer, 1998.

\bibitem{jaakkola1993convergence}
Tommi Jaakkola, Michael Jordan, and Satinder Singh.
\newblock Convergence of stochastic iterative dynamic programming algorithms.
\newblock {\em Advances in neural information processing systems}, 6, 1993.

\bibitem{ketkar2017introduction}
Nikhil Ketkar.
\newblock Introduction to keras.
\newblock In {\em Deep learning with Python}, pages 97--111. Springer, 2017.

\bibitem{kozuno2021model}
Tadashi Kozuno, Pierre M{\'e}nard, R{\'e}mi Munos, and Michal Valko.
\newblock Model-free learning for two-player zero-sum partially observable
  {M}arkov games with perfect recall.
\newblock {\em arXiv preprint arXiv:2106.06279}, 2021.

\bibitem{lemke1964equilibrium}
Carlton~E Lemke and Joseph~T Howson, Jr.
\newblock Equilibrium points of bimatrix games.
\newblock {\em Journal of the Society for industrial and Applied Mathematics},
  12(2):413--423, 1964.

\bibitem{littman2001value}
Michael~L Littman.
\newblock Value-function reinforcement learning in {M}arkov games.
\newblock {\em Cognitive systems research}, 2(1):55--66, 2001.

\bibitem{littman2001friend}
Michael~L Littman et~al.
\newblock Friend-or-foe {Q}-learning in general-sum games.
\newblock In {\em ICML}, volume~1, pages 322--328, 2001.

\bibitem{DBLP:journals/corr/abs-2003-02372}
Jieliang Luo and Hui Li.
\newblock Dynamic experience replay.
\newblock {\em CoRR}, abs/2003.02372, 2020.

\bibitem{nash1951non}
John Nash.
\newblock Non-cooperative games.
\newblock {\em Annals of mathematics}, pages 286--295, 1951.

\bibitem{rust2008dynamic}
John Rust.
\newblock Dynamic programming.
\newblock {\em The new Palgrave dictionary of economics}, 1:8, 2008.

\bibitem{shamma2004unified}
Jeff~S Shamma and G{\"u}rdal Arslan.
\newblock Unified convergence proofs of continuous-time fictitious play.
\newblock {\em IEEE Transactions on Automatic Control}, 49(7):1137--1141, 2004.

\bibitem{southey2012bayes}
Finnegan Southey, Michael~P Bowling, Bryce Larson, Carmelo Piccione, Neil
  Burch, Darse Billings, and Chris Rayner.
\newblock Bayes' bluff: Opponent modelling in poker.
\newblock In {\em UAI'05: Proceedings of the Twenty-First Conference on
  Uncertainty in Artificial Intelligence}, page 550–558, July 2005.

\bibitem{sutton2018reinforcement}
Richard~S Sutton and Andrew~G Barto.
\newblock {\em Reinforcement learning: An introduction}.
\newblock MIT press, 2018.

\bibitem{szepesvari1999unified}
Csaba Szepesv{\'a}ri and Michael~L Littman.
\newblock A unified analysis of value-function-based reinforcement-learning
  algorithms.
\newblock {\em Neural computation}, 11(8):2017--2060, 1999.

\bibitem{wang2002reinforcement}
Xiaofeng Wang and Tuomas Sandholm.
\newblock Reinforcement learning to play an optimal {N}ash equilibrium in team
  {M}arkov games.
\newblock {\em Advances in neural information processing systems}, 15, 2002.

\bibitem{https://doi.org/10.48550/arxiv.2109.06668}
Tianpei Yang, Hongyao Tang, Chenjia Bai, Jinyi Liu, Jianye Hao, Zhaopeng Meng,
  Peng Liu, and Zhen Wang.
\newblock Exploration in deep reinforcement learning: A comprehensive survey,
  2021.

\bibitem{zha2019rlcard}
Daochen Zha, Kwei-Herng Lai, Yuanpu Cao, Songyi Huang, Ruzhe Wei, Junyu Guo,
  and Xia Hu.
\newblock Rlcard: A toolkit for reinforcement learning in card games.
\newblock In {\em IJCAI'20: Proceedings of the Twenty-Ninth International Joint
  Conference on Artificial Intelligence}, number 764, page 5264–5266, January
  2021.

\bibitem{zhang2021monte}
Li~Zhang, Yuxuan Chen, Wei Wang, Ziliang Han, Shijian Li, Zhijie Pan, and Gang
  Pan.
\newblock A monte carlo neural fictitious self-play approach to approximate
  nash equilibrium in imperfect-information dynamic games.
\newblock {\em Frontiers of Computer Science}, 15(5):1--14, 2021.

\end{thebibliography}
\bibliographystyle{plain}

\end{document}